\documentclass[a4paper]{article}
\usepackage[utf8]{inputenc}
\usepackage{geometry}
\usepackage{amsmath}
\usepackage{amssymb}
\usepackage{amsthm}
\usepackage{mathtools}
\usepackage{bbm}
\usepackage[all]{xy}
\usepackage{hyperref}

\theoremstyle{definition}
\newtheorem{defn}{Definition}
\newtheorem{rem}[defn]{Remark}
\theoremstyle{plain}
\newtheorem{teo}[defn]{Theorem}
\newtheorem{cor}[defn]{Corollary}
\newtheorem{lem}[defn]{Lemma}

\DeclareMathOperator{\Hom}{Hom}
\DeclareMathOperator{\End}{End}
\DeclareMathOperator{\Der}{Der}
\DeclareMathOperator{\DR}{DR}
\DeclareMathOperator{\Rep}{Rep}
\DeclareMathOperator{\Mat}{Mat}
\DeclareMathOperator{\GL}{GL}
\DeclareMathOperator{\PGL}{PGL}
\DeclareMathOperator{\tr}{tr}

\newcommand{\app}[3]{#1\colon #2\to #3}
\newcommand{\bb}[1]{\mathbbm{#1}}

\newcommand{\de}{\mathrm{d}}
\newcommand{\deq}{\mathrel{:=}}
\newcommand{\falg}[2]{#1\langle #2\rangle}
\newcommand{\gc}[2]{[\![#1,#2]\!]}
\newcommand{\ld}[1]{\mathcal{L}_{#1}}
\newcommand{\ol}[1]{\overline{#1}}
\newcommand{\pair}[2]{\langle #1,#2\rangle}
\newcommand{\bd}[1]{\mathrm{T}{#1}}

\begin{document}

\title{Poisson-Nijenhuis structures on quiver path algebras}
\author{Claudio Bartocci and Alberto Tacchella}
\maketitle

\begin{abstract}
  We introduce a notion of noncommutative Poisson-Nijenhuis structure on the
  path algebra of a quiver. In particular, we focus on the case when the
  Poisson bracket arises from a noncommutative symplectic form. The formalism
  is then applied to the study of the Calogero-Moser and Gibbons-Hermsen
  integrable systems. In the former case, we give a new interpretation of the
  bihamiltonian reduction performed in \cite{bfmop10}.
\end{abstract}

\section{Introduction}

Since Magri's seminal paper \cite{mag78}, the notion of bihamiltonian manifold
has played a central role in the theory of integrable systems. Some of the
most significant examples of bihamiltonian manifolds arise from a
Poisson-Nijenhuis (PN) structure. We briefly recall that a PN structure on a
differentiable manifold $M$ is a pair $(\pi_0, N)$, where $\pi_0$ is a Poisson
bivector on $M$ and $N$ is an endomorphism of the tangent bundle $TM$ whose
Nijenhuis torsion vanishes and which satisfies a suitable compatibility
condition with $\pi_0$ \cite{mm84}. With these ingredients one may introduce a
second Poisson bivector, $\pi_1= \pi_0\circ N$, such that $[\pi_0, \pi_1] =
0$, where $[\cdot, \cdot]$ is the Schouten bracket on polyvector fields. In a
number of important cases --- e.g. for the Calogero-Moser system
\cite{bfmop10} --- the manifold $M$ is a cotangent bundle, $M = T^{*}X$,
$\pi_0$ is the inverse of the canonical symplectic form on $M$, and the
recursion operator $N = \pi_1\circ \pi_0^{-1}$ turns out be the \emph{complete
  lift} of a torsionless endomorphism $\app{L}{TX}{TX}$ \cite{tur92}.

The notion of Poisson bracket has been recently generalized to a
noncommutative geometric setting along the lines of the general approach
introduced by Kontsevich \cite{kont93} and developed by Ginzburg \cite{ginz01}
and other authors in the symplectic case. In particular, a notion of
\emph{double Poisson structure} on a general associative noncommutative
algebra $A$ has been introduced by van den Bergh \cite{vdb08}. When $A$ is the
path algebra of a quiver an alternative, yet equivalent, definition has been
proposed by Bielawski in the paper \cite{bie13}, where many explicit examples
are discussed. Double Poisson structures on free associative algebras have
been studied by Odesskii, Rubtsov and Sokolov \cite{ors13}, focusing in
particular on linear and quadratic structures.

In this paper we make a further step in this direction by introducing and
studying noncommutative Poisson-Nijenhuis structures on the path algebra $A$
of a quiver $Q$. As well known, the algebra of noncommutative differential
forms on $A$ can be defined according to a universal construction valid for
any associative algebra \cite{kar86, lod92}. On the other hand, a convenient
notion of polyvector fields on $A$ has been introduced in \cite{bie13}: in
this formalism a double Poisson structure on $Q$ is equivalent to the
assignment of a bivector $\pi$ such that $[\pi, \pi]=0$ (see for details
\S\S~\ref{gen}, \ref{diff}). The delicate issue is then to devise an
appropriate definition of tensors of type $(1,1)$, in order to have
``recursion operators'' as in the commutative setting (def.~\ref{def:reg}).

Once a Poisson bivector $\pi$ and a recursion operator $N$ on the path algebra
$A$ are given, one may mimic the classical theory of PN manifolds by
noticing that all relevant results can be proved in the purely algebraic
language of Lie algebras and their deformations \cite{ksm90}. Along these
lines, we are able to obtain theorem \ref{mainth1} generalizing the result
concerning the existence of a hierarchy of compatible Poisson structures on
any PN manifold. When the Poisson bivector arises from a noncommutative
symplectic structure in the sense of \cite{ginz01}, we prove theorem
\ref{compat-sympl}, which extends the usual result for $\omega N$ manifolds.
Furthermore, we are able to recover, in our environment, the above mentioned
construction of $\omega N$ manifolds through the complete lift of an
endomorphism of the tangent bundle (\S~\ref{lifts}).

In section \ref{sec:ex} we discuss two significant applications of our
formalism, namely the noncommutative versions of the rational Calogero-Moser
system and of the Gibbons-Hermsen system. The bihamiltonian structure of the
Calogero-Moser system was first described in \cite{mm96}; more recently, a
geometric interpretation of that structure was given in \cite{bfmop10} by
means of a two-step reduction of the two Poisson bivectors. The path algebra
of the quiver with two loops provides the natural noncommutative counterpart
of Calogero-Moser phase space, as shown in \cite{ginz01}. In \S~\ref{CM} we
define a noncommutative $\omega N$ structure on this path algebra and prove
that it induces --- first on the representation spaces, then on the quotient
space --- the $\omega N$ structures used in \cite{bfmop10}.

The Gibbons-Hermsen system \cite{gh84} is a generalization of Calogero-Moser
but, up to our knowledge, no bihamiltonian structure for it is known. As a
noncommutative counterpart of the rank 2 Gibbons-Hermsen phase space we take
the path algebra of the double of the quiver \eqref{ghquiver} already studied
by Bielawski and Pidstrygach in \cite{bp08}. In \S~\ref{GH} we construct a
noncommutative $\omega N$ structure on this path algebra and obtain the
corresponding bihamiltonian hierarchy. We expect that a corresponding $\omega
N$ structure is induced not only on the representation spaces of the quiver
\eqref{ghquiver}, but also on the phase space of the system (conjecturally to
be defined along the same guidelines as in \cite{bfmop10}). Finally, in
section \ref{sec:fin} we speculate briefly about other possible developments
of the ideas presented in this paper.

In the remainder of this introduction we set up our notation for quivers and
quiver representations (for this matter our basic reference is \cite{brion}).

\subsection{Quivers and their representations}
\label{notation-quiver}

A quiver $Q$ is a finite oriented graph. We think of $Q$ as the (finite) set
of its arrows; the (finite) set of vertices of $Q$ will be denoted by $I$ and
its element will be labeled by $e_1,\dots e_n$. One has maps
$\app{h,t}{Q}{I}$ which associate to each arrow its head and tail. The double
of $Q$ is the quiver $\ol{Q}$ obtained by attaching, for each arrow $a$ in
$Q$, a dual arrow $a^\ast$ with the same endpoints but with opposite
direction, that is $t(a^\ast) = h(a)$, $h(a^\ast) = t(a)$.

Let \(\bb{k}\) be a field of characteristic zero. The path algebra $\bb{k}Q$
is the associative algebra over $\bb{k}$ generated by the paths in $Q$
(including the trivial ones) with product given by concatenation of paths
whenever is possible, zero otherwise. Clearly, the arrows $\{a\}_{a\in Q}$ and
the trivial paths, identified with the vertices $e_1,\dots e_n$, are a set of
generators for $\bb{k}Q$. If $h(a) = t(b)$, we shall write $ba$ for the
resulting concatenated path; observe that $e_{h(a)} a = a e_{t(a)} = a$ for
all $a\in Q$. 

Let $B$ denote the commutative semisimple algebra $\bigoplus_{i=1}^n
\bb{k}e_i$, where the $e_i$ are orthogonal idempotents, i.e. $e_i^2=0$ and
$e_i e_j = 0$ for $i\neq j$. There is a natural algebra embedding
$B\hookrightarrow \bb{k}Q$ which gives $\bb{k}Q$ a structure of $B$-algebra.
 
A $\bb{k}$-representation of a quiver $Q$ is a pair $(V, \tau)$, where $V=
\bigoplus_{i\in I} V_{i}$ is an $I$-graded $\bb{k}$-vector space and
$\tau=(\tau_{a})_{a\in Q}$ is a set of linear maps
$\tau_{a}\in\Hom_{\bb{k}}(V_{t(a)},V_{h(a)})$. The space of the
representations of $Q$ on $V$ will be denoted by $\Rep_{\bb{k}}({Q},V)$.

Let us write $\app{\pi_i}{V}{V_i}$ for the canonical projection onto $V_i$ and
$\app{\jmath_i}{V_i}{V}$ for the canonical immersion of $V_i$. Then each
$\tau_a$ determines an element $\tilde{\tau}_a \in \End(V)$ given by
$\tilde{\tau}_a = \jmath_a \tau_a \pi_a$; similarly, for each vertex $e_i$ we
define $\tilde{\tau}_i \in \End(V)$ as the composition $\tilde{\tau}_i =
\jmath_i \pi_i$. It is straightforward to verify that these endomorphisms
satisfy the relations
$$
\tilde{\tau}_i ^2 = \tilde{\tau}_i\,; \quad \tilde{\tau}_i \tilde{\tau}_j = 0
\ \text{for}\ i\neq j\,; \quad \tilde{\tau}_{h(a)} \tilde{\tau}_a =
\tilde{\tau}_a \tilde{\tau}_{t(a)} = \tilde{\tau}_a\,.
$$ 
The algebra $\bigoplus_{i=1}^n \bb{k} \tilde{\tau}_i$ may be identified with
$B$. Each representation $(V, \tau)$ induces a $B$-algebra homomorphism
$\bb{k}Q \to \End(V)$ defined by $a \mapsto \tilde{\tau}_a$,
$e_i \mapsto \tilde{\tau}_i$, and, conversely, each such a homomorphism
determines a representation of $Q$ on $V$. Summing up, one has an isomorphism
\begin{equation}
  \label{repr-isomo}
  \mathcal{R}\colon\Rep_{\bb{k}}({Q},V) \stackrel{\sim}{\longrightarrow}
  \Hom_{B{\textrm{-alg}}} \left(\bb{k}Q, \End(V)\right)\,.
\end{equation}
Let us fix an element $\mathbf{n}=(n_i)_{i\in I}\in\mathbb{N}^I$ and set
$\vert \mathbf{n}\vert = \sum_i n_i$. The space of representations of the
quiver $Q$ on $\bb{k}^{\vert \mathbf{n}\vert} = \bigoplus_{i\in I}
\bb{k}^{n_i}$ will denoted by $\Rep_{\bb{k}}({Q},\mathbf{n})$:
\begin{equation}
  \label{eq:decomp}
  \Rep_{\bb{k}}({Q},\mathbf{n}) = \bigoplus_{a\in Q} \Mat_{n_{h(a)}\times n_{t(a)}} (\bb{k})\,.
\end{equation}
In this case the map \eqref{repr-isomo} becomes
\begin{equation}
  \begin{aligned}
    \label{repr-isomo2}
    \mathcal{R}\colon\Rep_{\bb{k}}({Q},\mathbf{n}) &\stackrel{\sim}{\longrightarrow}
    \Hom_{B{\textrm{-alg}}} \left(\bb{k}Q, \Mat_{\vert \mathbf{n}\vert\times \vert \mathbf{n}\vert} ({\bb{k}}\right) \\
    \tau &\longmapsto \mathcal{R}(\tau)\,,
  \end{aligned}
\end{equation}
where $\mathcal{R}(\tau)(a) = \tilde{\tau}_a$ for all $a\in Q$ and
$\mathcal{R}(\tau)(e_i) = \tilde{\tau}_i$ for all trivial paths $e_i$.
Clearly, if we give a matrix $R \in \Mat_{\vert \mathbf{n}\vert\times \vert
  \mathbf{n}\vert}({\bb{k}})$ the block decomposition $R = R_{ij}$, with
$R_{ij} \in \Mat_{n_i \times n_j} ({\bb{k}})$, the only non-zero block of
$\tilde{\tau}_a$  is the $n_{h(a)}\times n_{t(a)}$ block corresponding to
$\tau_a$, and the only non-zero block of $\tilde{\tau}_i$ is the $n_i \times
n_i$ identity matrix. 

The group 
\begin{equation}
  \label{eq:gl-bfn}
  \GL_{\mathbf{n}}(\bb{k})\deq \prod_{i\in I} \GL_{n_i}(\bb{k})
\end{equation}
acts naturally on $\Rep_{\bb{k}}({Q},\mathbf{n})$ by conjugation and preserves
the decomposition \eqref{eq:decomp}. The subgroup
\[ \bb{k}^{*}I_{\mathbf{n}} = \{(\lambda I_{n_{i}})_{i\in I}\mid \lambda\in \bb{k}^{*}\} \]
is contained in the center of \(\GL_{\mathbf{n}}(\bb{k})\) and acts trivially
on \(\Rep_{\bb{k}}(Q,\mathbf{n})\). Thus the action on
\(\GL_{\mathbf{n}}(\bb{k})\) factors through an action of the group
\begin{equation}
  \label{eq:group}
  G_{\mathbf{n}}\deq \GL_{\mathbf{n}}(\bb{k})/\bb{k}^{*}I_{\mathbf{n}}\,.
\end{equation}
The isomorphism classes of representations of the quiver $Q$ with a fixed
dimension vector $\mathbf{n} = (\dim V_i)_{i\in I}$ are then in one to one
correspondence with the set of orbits of $G_\mathbf{n}$ in
$\Rep_{\bb{k}}({Q},\mathbf{n})$.

\section{Non-commutative PN structures}
\label{sec:pn}

\subsection{General setting}
\label{gen}

In order to develop a noncommutative PN formalism on quiver path algebras we
need to briefly recall some basic notions (see also \cite{nc1} for a more
pedagogical introduction).

Let $A$ be a noncommutative, associative, unital algebra over a field $\bb{k}$
of characteristic zero. The definition of the differential graded (DG) algebra
of differential forms on $A$ dates back to the classical work of A.~Connes and
M.~Karoubi in the mid 1980s \cite{kar86, lod92, cq95}. Let $\tilde A$ be the
quotient vector space $A/ \bb{k}$ and define
$$
\Omega_{\bb{k}}^r (A) = A \otimes_{\bb{k}}
\underbrace{{\tilde A}\otimes_{\bb{k}} \cdots \otimes_{\bb{k}} {\tilde A}}_{r\text{ times}}
$$ 
for any integer $r\geq 0$. The graded vector space 
$\Omega_{\bb{k}}^{\bullet} (A) = \bigoplus_{r\geq 0} \Omega_{\bb{k}}^r (A)$ is
endowed with the graded product
\begin{equation}
  \label{eq:product}
  [a_0\otimes a_1\otimes \cdots \otimes a_r] [a_{r+1}\otimes \cdots \otimes a_{s}] 
  = \sum_{i=0}^{r} (-1)^{r-i} [a_0\otimes \cdots \otimes a_i a_{i+1}\otimes
  \cdots \otimes a_s]\,,
\end{equation}
where $[a_0\otimes a_1\otimes \cdots \otimes a_r]$ is the class of $a_0\otimes
a_1\otimes \cdots \otimes a_r$ in $\Omega_{\bb{k}}^r (A)$, and with the
differential
\begin{equation}
  \label{eq:differential}
  \de [a_0\otimes a_1\otimes \cdots \otimes a_r] = [1\otimes a_0\otimes
  a_1\otimes \cdots \otimes a_r]\,.
\end{equation}
It is not difficult to show that these formulas determine the unique DG
algebra structure on $\Omega_{\bb{k}}^{\bullet} (A)$ satisfying the condition
$$
[a_0\otimes a_1\otimes \cdots \otimes a_r] = a_0 \de a_1 \cdots \de a_r\,.
$$ 
The mapping $a_0 \de a_1 \mapsto a_0\otimes a_1 - a_0 a_1\otimes 1$ yields a
natural isomorphism $\Omega_{\bb{k}}^1 (A) \stackrel{\sim}{\longrightarrow}
\ker\mu$, where $\app{\mu}{A\otimes_{\bb{k}} A}{A}$ is the multiplication
morphism. In this way $\Omega_{\bb{k}}^1 (A)$ can be given a structure of
$A$-bimodule; while the left multiplication is the obvious one, the right
multiplication is somewhat less evident: $(a_0\de a_1) a = a_0\de (a_1 a) -
a_0 a_1\de a$.

The derivation functor $\app{\Der_{\bb{k}}(A, \cdot)}{A\textrm{\bf -Bimod}}
{\textrm{\bf Vect}_{\bb{k}}}$ is representable by $\Omega_{\bb{k}}^1 (A)$. So,
for any $A$-bimodule $M$, there is an isomorphism 
\begin{equation*}
  \Der_{\bb{k}}(A, M) \stackrel{\sim} {\longrightarrow}
  \Hom_{A{\textrm{-Bimod}}} (\Omega_{\bb{k}}^1 (A), M)\,.
\end{equation*}
When $M=A$ this isomorphism induces a pairing 
\begin{equation}
  \begin{split}
    \Omega_{\bb{k}}^1 (A) \times \Der_{\bb{k}}(A,A) &\to A\\
    (\alpha, \theta) &\mapsto  i_{\theta}(\alpha)
  \end{split}
\end{equation}
Notice that, since the linear space \(\Der_{\bb{k}}(A,A)\) has no natural
structure of \(A\)-bimodule\footnote{In general, \(\Der_{\bb{k}}(A,A)\) is
  only a \(Z(A)\)-bimodule, \(Z(A)\) being the center of \(A\). For quiver
  path algebras one has \(Z(A) = \bb{k}\).}, this is just a pairing between
vector spaces over \(\bb{k}\). For any derivation $\theta\in
\Der_{\bb{k}}(A,A)$ the operation $i_{\theta}$ extends to the whole of
$\Omega_{\bb{k}}^{\bullet} (A)$:
$$
i_{\theta} (a_0 \de a_1 \cdots \de a_r) = \sum_{j=1}^r (-1)^{j-1} a_0 \de a_1
\cdots i_{\theta}(a_j)\cdots \de a_r\,.
$$
The Lie derivative $\app{\ld{\theta}}{\Omega_{\bb{k}}^{\bullet}(A)}
{\Omega_{\bb{k}}^{\bullet} (A)}$ with respect to $\theta$ may then be defined
using the Cartan formula $\ld{\theta} = \de\circ i_{\theta} + i_{\theta}
\circ \de$. It follows that any Lie derivative $\ld{\theta}$ is a degree zero
derivation of $\Omega_{\bb{k}}^{\bullet} (A)$, and the following identities
are readily verified on $\Omega_{\bb{k}}^1 (A)$ (and therefore on the whole of
$\Omega_{\bb{k}}^{\bullet} (A)$):
\begin{equation}
  \label{eq:id-ld}
  [\ld{\theta}, \ld{\eta}] = \ld{[\theta,\eta]}\,,\qquad
  [\ld{\theta}, i_\eta]= i_{[\theta,\eta]}\,,
\end{equation}
where $[X, Y] = X\circ Y - Y\circ X$ is the usual commutator of endomorphisms.

The DG algebra $\Omega_{\bb{k}}^{\bullet} (A)$ comes naturally equipped with
the graded commutator
$$
\gc{\chi}{\omega} = \chi\omega - (-1)^{\vert\chi\vert \vert\omega\vert} \omega\chi\,.
$$
The abelianization of $\Omega_{\bb{k}}^{\bullet} (A)$ is the graded vector
space
$$
\DR_{\bb{k}}^{\bullet}(A)\deq \Omega_{\bb{k}}^{\bullet}(A)/
\gc{\Omega_{\bb{k}}^{\bullet} (A)}{\Omega_{\bb{k}}^{\bullet} (A)}\,,
$$
where $\gc{\Omega_{\bb{k}}^{\bullet} (A)}{\Omega_{\bb{k}}^{\bullet} (A)}$ is
the linear subspace generated by all graded commutators. The differential
\eqref{eq:differential} descends to this quotient and so one gets a complex
$(\DR_{\bb{k}}^{\bullet}(A), \de)$, whose cohomology is, by definition, the
noncommutative de Rham cohomology of $A$. Notice that, being every element of
$A$ of degree zero, one has $\gc{A}{A} = [A,A]$ so that the degree zero term
of this complex is the linear space $\DR_{\bb{k}}^0(A) = A/[A,A]$, to be
interpreted as the space of ``regular functions'' associated to the algebra
$A$. Similarly, the degree one term is $\DR_{\bb{k}}^1(A) = \Omega_{\bb{k}}^1
(A)/[A, \Omega_{\bb{k}}^1 (A)]$.

It is easy to verify that, for any derivation $\theta\in \Der_{\bb{k}}(A,A)$,
the operations $i_{\theta}$ and $\mathcal{L}_{\theta}$ induce operations,
denoted by the same symbols, on the complex $\DR_{\bb{k}}^{\bullet}(A)$. We
can therefore define a linear pairing $\app{\pair{\cdot}{\cdot}}
{\DR^{1}_{\bb{k}}(A)\times \Der_{\bb{k}}(A,A)}{\DR^{0}_{\bb{k}}(A)}$ given by
\begin{equation}
  \label{eq:def-pair-nc}
  \pair{\alpha}{\theta} = i_{\theta}(\alpha) \mod [A,A]\,.
\end{equation}
Whenever a subalgebra $B \hookrightarrow A$ is assigned, all previous
constructions can be performed relatively to $B$. Specifically, one sets
$$
\Omega_{B}^r (A) = A \otimes_{B} \underbrace{{A/B}\otimes_B \cdots \otimes_B
  {A/B}}_{r \text{ times}}\,,\qquad  \Omega_{B}^\bullet (A) = \bigoplus_{r\geq
  0}\Omega_{B}^r (A)
$$
and checks that the formulas \eqref{eq:product}, \eqref{eq:differential}
descend to $\Omega_{B}^\bullet (A)$ and endow it with a structure a DG
algebra. The vector space $\Omega_{B}^1 (A)$ is isomorphic to the kernel of
the multiplication morphism $A\otimes_B A \to A$ (thus inheriting a structure
of $A$-bimodule) and represents the derivation functor $\app{\Der_B (A,\cdot)}
{A\textrm{\bf -Bimod}}{\textrm{\bf Vect}_{\bb{k}}}$. The relative de Rham
complex of $A$ is then defined as the quotient
$$
\DR_{B}^{\bullet}(A) = \Omega_{B}^{\bullet} (A)/ \gc{\Omega_{B}^{\bullet}
  (A)}{\Omega_{B}^{\bullet} (A)}
$$
and, as expected, one has a pairing
\begin{equation}
  \label{eq:def-pair-nc2}
  \app{\pair{\cdot}{\cdot}}{\DR^{1}_{B}(A)\times \Der_{B}(A,A)}{\DR^{0}_B (A)}\,.
\end{equation}

\subsection{Differential calculus on path algebras}
\label{diff}

From now on we shall restrict our attention to the case when $A$ is the path
algebra $\bb{k}Q$ of a quiver $Q$ and $B=\bigoplus_{i\in I} \bb{k} e_i$ is its
commutative subalgebra of idempotents. To make the notation less cumbersome,
we shall adopt the following abbreviations:
$$
\Omega^\bullet (Q)\deq \Omega^\bullet_B (\bb{k}Q)\,,\quad 
\Der(Q)\deq \Der_B(\bb{k}Q, \bb{k}Q)\,,\quad \DR^\bullet(Q)\deq
\DR^\bullet_B(\bb{k}Q)\,.
$$
Following R.~Bielawski's approach \cite{bie13}, we denote each dual arrow
$a^\ast$ of the double quiver $\ol{Q}$ by $\partial_a$ and think of it as a
fundamental noncommutative vector field. To emphasize this different
interpretation of $\ol{Q}$ we adopt a new symbol to denote it: $\bd{Q}$.

Let us consider the linear subspace ${\bb{k}\bd{Q}}^r \subset \bb{k}\bd{Q}$
generated by all the monomials $x_1\cdots x_k$ with $k\geq r$ such that
exactly $r$ of the $x_i$ are of the type $\partial_a$ for some $a\in Q$.
Obviously, one has ${\bb{k}\bd{Q}}^0 = \bb{k}Q$. The vector space
$\bb{k}\bd{Q}$ can therefore be given the grading
\begin{equation}
  \label{grading-double}
  \bb{k}\bd{Q} = \bigoplus_{r\geq 0} {\bb{k}\bd{Q}}^r\,.
\end{equation}
\begin{defn}
  The space $\mathcal{V}Q$ of \emph{noncommutative polyvector fields} on the
  quiver $Q$ is the quotient of $\bb{k}\bd{Q}$ by the relations
  \begin{equation}
    \label{relations-polyvectors}
    PR - (-1)^{pr} RP = 0\,, \quad \text{if}\ P\in {\bb{k}\bd{Q}}^p\,, R\in {\bb{k}\bd{Q}}^r\,.
  \end{equation}
\end{defn}
It is worth observing that every path which is not closed becomes zero in
$\mathcal{V}Q$; in other words, $\mathcal{V}Q$ is generated by closed paths
(``necklaces''). The grading (\ref{grading-double}) induces a grading on
$\mathcal{V}Q$, i.e. $\mathcal{V}Q = \bigoplus_{r\geq 0} \mathcal{V}^r Q$.
Notice that $\mathcal{V}^0 Q = \DR^0(Q)$. As for $\mathcal{V}^1Q$, its
elements can be written in the canonical form
\begin{equation}
  \label{can-form-deriv}
  \theta = \sum_{a\in Q} p_{a}\partial_{a}\,, \quad \text{with }
  p_a\in \bb{k}Q,\  e_{h(a)}p_a = p_a,\ p_ae_{t(a)} = p_a\,.
\end{equation}
\begin{lem}
  \label{der-v1}
  There is a canonical isomorphism $\mathcal{V}^1Q \simeq \Der(Q)$.
\end{lem}
\begin{proof}
  Each element $\theta\in \mathcal{V}^1Q$ of the form \eqref{can-form-deriv}
  uniquely determines a \(B\)-linear derivation $A\to A$ defined by mapping
  each arrow $a$ to the path $p_a$ and each idempotent $e_{i}$ to zero.
\end{proof}
A canonical form is also available (see e.g. \cite{blb02}) for every 1-form
$\alpha \in \DR^1(Q)$:
\begin{equation}
  \label{can-form-1forms}
  \alpha = \sum_{a\in Q} r_{a}\de a\,, \quad \text{with }
  r_a\in \bb{k}Q,\  e_{t(a)}r_a = r_a,\ r_ae_{h(a)} = r_a\,.
\end{equation}
Using expressions \eqref{can-form-deriv} and \eqref{can-form-1forms} the
pairing $\app{\pair{\cdot}{\cdot}}{\DR^1(Q) \times \Der(Q)}{\DR^0(Q)}$
introduced in equation \eqref{eq:def-pair-nc2} becomes simply
\begin{equation}
  \label{eq:def-pair-nc3}
  \pair{\alpha}{\theta} = \sum_{a\in Q} r_{a}p_{a}\,.
\end{equation}
This pairing is ``perfect'' in the sense that $\pair{\de a}{\partial_{b}} =
\delta_{ab}$ (but notice that both \(\Der(Q)\) and \(\DR^{1}(Q)\) are actually
infinite-dimensional linear spaces over \(\bb{k}\)).

The space $\mathcal{V}Q$ of noncommutative polyvector fields can be endowed
with a Schouten bracket \cite{bie13, pich08, laz05}. For any arrow $y\in
\bd{Q}$ and for any monomial $x_1\cdots x_N$, with $x_i \in {\bb{k}} \bd{Q}$,
let
\begin{equation}
  D_y(x_1\cdots x_N) = \sum_{x_i = y} (-1)^{n_i m_i} x_{i+1}\cdots x_N x_{1}\cdots x_{i-1}\,,
\end{equation}
where $n_i$ (resp.~$m_i$) is the number of dual arrows $\partial_a$ among the
elements $x_1, \dots, x_i$ (resp.~among $x_{i+1}, \dots, x_N$). This operation
can be extended linearly to the whole of ${\bb{k}}\bd{Q}$, so obtaining a
directional superderivative
$$
\app{D_y}{\mathcal{V}Q}{\bb{k}}\bd{Q}\,.
$$
\begin{defn}
  \label{schouten}
  Given $\lambda \in \mathcal{V}^p Q$, $\xi \in \mathcal{V}^q Q$, their
  \emph{Schouten bracket} $[\lambda, \xi]$ is defined by the formula
  $$
  [\lambda, \xi] = \sum_{a\in Q} D_{\partial_a}(\lambda) D_{a}(\xi) -
  (-1)^{(p+1)(q+1)} D_{\partial_a}(\xi) D_{a}(\lambda) \qquad\text{modulo
    relations \eqref{relations-polyvectors}.}
  $$
\end{defn}
For any $\lambda \in \mathcal{V}^p Q$, $\xi \in \mathcal{V}^q Q$, $\sigma\in
\mathcal{V}^r Q$, the following properties hold true:
\begin{enumerate}
\item[1)] $[\lambda, \xi] \in \mathcal{V}^{p+q-1} Q$;
\item[2)] $[\lambda, \xi] = -  (-1)^{(p+1)(q+1)}[\xi, \lambda]$;
\item[3)] (graded Jacobi identity)
  $$ [\lambda, [\xi, \sigma]] + (-1)^{(p+1)(q+1)} [\xi, [\sigma, \lambda]] + (-1)^{(q+1)(r+1)} [\sigma, [\lambda, \xi]] = 0\,.$$
\end{enumerate}
Let us now consider, for a given dimension vector $\mathbf{n}$, the
representation space $\Rep_{\bb{k}} ({Q},\mathbf{n})$ of the quiver $Q$ and
denote its space of $G_{\mathbf{n}}$-invariant differential forms by
$\Omega^{\bullet}(\Rep_{\bb{k}} ({Q},\mathbf{n}))^{G_{\mathbf{n}}}$ and that
of $G_{\mathbf{n}}$-invariant ordinary polyvector fields by
$\mathcal{V}(\Rep_{\bb{k}} ({Q},\mathbf{n}))^{G_{\mathbf{n}}}$ (the group
$G_{\mathbf{n}}$ is defined in eq.~\eqref{eq:group}). The space
$\mathcal{V}(\Rep_{\bb{k}} ({Q},\mathbf{n}))^{G_{\mathbf{n}}}$ comes equipped
with the bracket induced by the usual Schouten bracket on the space
$\mathcal{V}(\Rep_{\bb{k}} ({Q},\mathbf{n}))$, namely
\begin{multline*}
  [X_1 \wedge \cdots \wedge X_p, Y_1\wedge \cdots\wedge Y_q] = \\
  = \sum_{i,j} (-1)^{i+j} [X_i, Y_j] \wedge X_1\wedge \cdots X_{i-1}\wedge X_{i+1}\wedge \cdots \wedge X_p\wedge
Y_1\wedge\cdots \wedge Y_{j-1}\wedge Y_{j+1}\wedge\cdots \wedge Y_q\,.
\end{multline*}
\begin{teo}
  \label{teor-discesa}
  Let $\Rep_{\bb{k}}({Q},\mathbf{n})$ be a representation space for the quiver $Q$.
  \begin{enumerate}
  \item[1)] There is a morphism of graded $B$-algebras
    $$
    \app{\mathbf{\hat{}}\ }{\DR^{\bullet}(Q)}{ \Omega^\bullet(\Rep_{\bb{k}}
      ({Q},\mathbf{n}))^{G_{\mathbf{n}}}}
    $$
    which commutes with the respective differentials;
  \item[2)] there is a morphism of graded $B$-algebras   
    $$
    \app{\mathbf{\check{}}\ }{\mathcal{V}Q}{\mathcal{V}(\Rep_{\bb{k}}
      ({Q},\mathbf{n}))^{G_{\mathbf{n}}}}
    $$
    which commutes with the respective Schouten brackets;
  \item[3)] for every $\alpha \in {\DR^{1}(Q)}$ and $\theta\in {\mathcal{V}}^1Q$ one has
    $$
    \widehat{\pair{\alpha}{\theta}} = \pair{\hat{\alpha}}{\check{\theta}}\,.
    $$
  \end{enumerate}
\end{teo}
\begin{proof}
  A proof of (1) can be found in V.~Ginzburg's \emph{Lectures}
  \cite[\S 12.6]{ginzlect} in the case of a general noncommutative associative
  algebra. Point (2) is proved in \cite{bie13}; point (3) is then
  straightforward in view of lemma \ref{der-v1}.
\end{proof}
The formalism we have set up makes it natural to state the following
definition \cite{bie13}.
\begin{defn}
  A \emph{double Poisson structure} on $Q$ is a noncommutative bivector $\pi
  \in \mathcal{V}^2Q$ such that $[\pi, \pi]=0$.
\end{defn}
As an immediate consequence of theorem \ref{teor-discesa} we get the following
result, which is crucial for the applications we shall describe in section
\ref{sec:ex}.
\begin{cor}
  If $\pi$ is a double Poisson structure on $Q$ then $\check{\pi}$ is a
  Poisson structure on each representation space \(\Rep_{\bb{k}}(Q,\mathbf{n})\).
\end{cor}
The previous corollary has a converse, which shows that a double Poisson
structure on $Q$ is completely determined by the family of all induced Poisson
structures on the representation spaces of $Q$.
\begin{teo}[Theorem 3.9 in \cite{bie13}]
  \label{teo:bie}
  Let \(\pi\in \mathcal{V}^{2}Q\). If $\check{\pi}$ is a Poisson structure on
  all representation spaces \(\Rep_{\bb{k}}(Q,\mathbf{n})\) then \(\pi\) is a
  double Poisson structure on \(\bb{k}Q\).
\end{teo}
In view of the sequel we need to clarify the relationship between
noncommutative bivectors, that is elements of \(\mathcal{V}^{2}Q\), and linear
maps \(\DR^{1}(Q)\to \Der(Q)\).

Without loss of generality we can write a bivector \(\pi\in \mathcal{V}^{2}Q\)
as a sum of ordinary (i.e. not graded) commutators of the form
\[ \pi = \sum_{a,b\in Q} \sum_{j\in J} [P^{ab}_{j}\partial_{a},R^{ab}_{j}\partial_{b}]\,, \]
where \(J\) is some finite set and \(P^{ab}_{j}\), \(R^{ab}_{j}\) are paths in
the quiver \(Q\) such that each resulting monomial is a closed path in \(\bd{Q}\). 
We define the corresponding map \(\app{\tilde{\pi}}{\DR^{1}(Q)}{\Der(Q)}\) in
the following way: given \(\alpha\in \DR^{1}(Q)\), let \(\sum_{c\in Q}
S_{c}\de c\) be a representative for \(\alpha\) in canonical form. Then
\[ \tilde{\pi}(\alpha)\deq \sum_{a,b\in Q} \sum_{j\in J} (P^{ab}_{j} S_{a}
R^{ab}_{j}\partial_{b} - R^{ab}_{j} S_{b} P^{ab}_{j}\partial_{a})\,. \]
This works because for each arrow \(a\) the path \(S_{a}\) runs in the
opposite direction to \(a\), hence it can always replace \(\partial_{a}\)
inside a word in \(\bb{k}\bd{Q}\) without making it zero.

We can relate the map \(\tilde{\pi}\) with the action of the bivector \(\pi\in
\mathcal{V}^{2}Q\) on a pair of 2-forms \(\alpha,\beta\in \DR^{1}(Q)\) by the
usual formula
\begin{equation}
  \label{eq:rel-pi-pit}
  \pi(\alpha,\beta) = \pair{\beta}{\tilde{\pi}(\alpha)} =
  i_{\tilde{\pi}(\alpha)}(\beta)\,.
\end{equation}
More explicitly, if \(\alpha\) is represented by \(\sum_{c} S_{c}\de c\) and
\(\beta\) by \(\sum_{c'} T_{c'}\de c'\) then
\[ \pi(\alpha,\beta) = \sum_{a,b\in Q} \sum_{j\in J} (P^{ab}_{j} S_{a}
R^{ab}_{j} T_{b} - R^{ab}_{j} S_{b} P^{ab}_{j} T_{a})\,.  \]
Now we would like to introduce the analogue of tensors of type \((1,1)\) on
\(\bb{k}Q\). We locate the salient feature of these objects in their ability
to be interpreted simultaneously as maps \(\app{N}{\Der(Q)}{\Der(Q)}\) and as
maps \(\app{N^{*}}{\DR^{1}(Q)}{\DR^{1}(Q)}\), related by the familiar equality
\begin{equation}
  \label{eq:rel-N-Nst}
  \pair{N^{*}(\alpha)}{\theta} = \pair{\alpha}{N(\theta)} \quad\text{ for
    every } \alpha\in \DR^{1}(Q),\, \theta\in \Der(Q)\,,
\end{equation}
where the pairing is defined by equation \eqref{eq:def-pair-nc3}. Not every
endomorphism of \(\Der(Q)\) has this property.
\begin{defn}
  \label{def:reg}
  A \(\bb{k}\)-linear endomorphism \(\app{N}{\Der(Q)}{\Der(Q)}\) is called
  \emph{regular} if there exists a derivation \(\app{\de^{N}}{\bb{k}Q}
  {\Omega^{1}(Q)}\) such that \(i_{\theta}\circ \de^{N} = N(\theta)\) for
  every \(\theta\in \Der(Q)\).
\end{defn}
Clearly the map \(\de^{N}\), if exists, is unique since for every arrow \(a\in
Q\) the 1-form \(\de^{N}a\) is completely determined by the paths
\((i_{\partial_{b}}(\de^{N}a))_{b\in Q}\), which by definition coincide with
\(N(\partial_{b})(a)\). It follows that to each regular endomorphism \(N\) we
can associate the unique morphism of \(\bb{k}Q\)-bimodules
\[ \app{N^{*}}{\Omega^{1}(Q)}{\Omega^{1}(Q)} \]
defined by sending each generator \(\de a\) of \(\Omega^{1}(Q)\) to the 1-form
\(\de^{N}a\). As every morphism of \(\bb{k}Q\)-bimodules preserves the linear
subspace \([\bb{k}Q,\Omega^{1}Q]\) inside \(\Omega^{1}(Q)\), this recipe
induces a unique \(\bb{k}\)-linear map \(\DR^{1}(Q)\to \DR^{1}(Q)\) that we
also denote by \(N^{*}\). By definition, we have
\[ \pair{N^{*}(\de a)}{\theta} = i_{\theta}(N^{*}(\de a)) =
i_{\theta}(\de^{N} a) = N(\theta)(a) = i_{N(\theta)}(\de a) = 
\pair{\de a}{N(\theta)}\,, \]
from which \eqref{eq:rel-N-Nst} follows by the \(\bb{k}Q\)-linearity of
\(N^{*}\) and \(i_{\theta}\). We shall call \(N^{*}\) the \emph{transpose} of
\(N\).

Let us remark that, conversely, if we are given a \(\bb{k}Q\)-linear map
\(\app{N^{*}}{\Omega^{1}(Q)}{\Omega^{1}(Q)}\) such that \eqref{eq:rel-N-Nst}
holds then the endomorphism \(N\) is necessarily regular, as the map
\(\app{\de^{N}}{\bb{k}Q}{\Omega^{1}(Q)}\) defined as the unique derivation
sending the arrow \(a\in Q\) to \(N^{*}(\de a)\) clearly has the property
required by definition \ref{def:reg}.

Hence for us a \emph{tensor of type (1,1) on} \(\bb{k}Q\) will be given by,
equivalently, a regular endomorphism \(\app{N}{\Der(Q)}{\Der(Q)}\) or the
corresponding transpose \(\app{N^{*}}{\DR^{1}(Q)}{\DR^{1}(Q)}\).

\subsection{PN structures on path algebras}

We continue by briefly recalling some concepts and results from \cite{ksm90}.
\begin{defn}
  \label{def:nt}
  Let \(E\) be a Lie algebra over the field \(\bb{k}\) and \(\app{N}{E}{E}\)
  be a linear map.
  \begin{enumerate}
  \item The \(N\)\emph{-deformed bracket} on \(E\) is the \(E\)-valued
    2-form on \(E\) defined by
    \[ [x,y]_{N}\deq [N(x),y] + [x,N(y)] - N([x,y])\,. \]
  \item The \emph{Nijenhuis torsion} of \(N\) is the \(E\)-valued 2-form on
    \(E\) defined by
    \begin{equation}
      \label{eq:def-TN}
      \mathcal{T}_{N}(x,y)\deq [N(x),N(y)] - N\bigl( [N(x),y] + [x,N(y)] -
        N([x,y])\bigr)\,.
    \end{equation}
  \end{enumerate}
  Two Lie brackets on \(E\) are \emph{compatible} if their sum is also a Lie
  bracket.
\end{defn}
\begin{teo}[Corollary 1.1 in \cite{ksm90}]
  If \(\mathcal{T}_{N}=0\) then \([\cdot,\cdot]_{N}\) is also a Lie bracket on
  \(E\) which is compatible with \([\cdot,\cdot]\).
\end{teo}
It also follows that \(N\) is a morphism of Lie algebras from
\((E,[\cdot,\cdot]_{N})\) to \((E,[\cdot,\cdot])\): indeed the condition
\(\mathcal{T}_{N}=0\) is equivalent to
\[ [N(x),N(y)] = N([x,y]_{N}) \quad\text{ for all } x,y\in E\,. \]
We can now give the following
\begin{defn}
  Let \(Q\) be a quiver. A \emph{Nijenhuis tensor on the path algebra}
  \(\bb{k}Q\) is a regular endomorphism \(\app{N}{\Der(Q)}{\Der(Q)}\) such
  that \(\mathcal{T}_{N}=0\).
\end{defn}
Let us recall how a Nijenhuis tensor on \(\bb{k}Q\) defines a ``deformed
version'' of the Cartan calculus on \(\DR^{\bullet}(Q)\). We have already
defined the derivation \(\app{\de^{N}}{\bb{k}Q}{\Omega^{1}(Q)}\), which can be
extended to a degree 1 derivation of the DG algebra \(\Omega^{\bullet}(Q)\) in
the usual way, that is imposing the (anti)commutation rule
\[ \de^{N}\circ \de + \de\circ \de^{N} = 0\,. \]
The vanishing of \(\mathcal{T}_{N}\) then implies \(\de^{N}\circ \de^{N}=0\).
We can also define a deformed Lie derivative \(\ld{\theta}^{N}\) as
\[ \ld{\theta}^{N}\deq \de^{N}\circ i_{\theta} + i_{\theta}\circ \de^{N}\,. \]
These operators obey a deformed version of the identities \eqref{eq:id-ld}
where the usual commutator bracket on \(\Der(Q)\) is replaced by the bracket
\([\cdot,\cdot]_{N}\):
\[ [\ld{\theta}^{N}, \ld{\eta}^{N}] = \ld{[\theta,\eta]_{N}}^{N}\,,\qquad
[\ld{\theta}^{N}, i_{\eta}]= i_{[\theta,\eta]_{N}}\,. \]
In particular the action of \(\ld{\theta}^{N}\) on 1-forms is given by
\[ \ld{\theta}^{N}(\beta) = \ld{N(\theta)}(\beta) - \ld{\theta}(N^{*}(\beta))
+ N^{*}(\ld{\theta}(\beta))\,. \]
Finally, both maps \(\de^{N}\) and \(\ld{\theta}^{N}\) descend from
\(\Omega^{\bullet}(Q)\) to \(\DR^{\bullet}(Q)\) by the usual arguments.

Suppose now that the quiver \(Q\) comes equipped with a noncommutative
bivector \(\pi\in \mathcal{V}^{2}Q\). Then we can ``dualize'' the Lie bracket
\([\cdot,\cdot]\) on \(\Der(Q)\) to a bracket defined on \(\DR^{1}(Q)\)
according to the well-known formula
\begin{equation}
  \label{eq:def-br-1f}
  \{\alpha,\beta\}_{\pi}\deq \ld{\tilde{\pi}(\alpha)}(\beta) -
  \ld{\tilde{\pi}(\beta)}(\alpha) - \de (\pi(\alpha,\beta))\,.
\end{equation}
\begin{teo}[Proposition 3.2 in \cite{ksm90}]
  The bracket \eqref{eq:def-br-1f} obeys the Jacobi identity if and only if
  \(\pi\) is a double Poisson structure, and in this case one has
  \begin{equation}
    \label{eq:pit-mor}
    \tilde{\pi}(\{\alpha,\beta\}_{\pi}) =
    [\tilde{\pi}(\alpha),\tilde{\pi}(\beta)]\,,
  \end{equation}
  that is, \(\tilde{\pi}\) is a morphism of Lie algebras.
\end{teo}
\begin{proof}
  The proof in \cite{ksm90} uses only the algebraic properties of the Schouten
  bracket, therefore it works also in our setting. Alternatively, one can
  simply observe that the analogous result holds for the induced (commutative)
  structures on every representation space, hence it must hold already at the
  noncommutative level.
\end{proof}
From now on we assume that the path algebra \(\bb{k}Q\) is equipped with both
a Nijenhuis tensor \(N\) and a double Poisson structure \(\pi\). Let us say
that \(N\) and \(\pi\) are \emph{algebraically compatible} if \(N\circ
\tilde{\pi} = \tilde{\pi}\circ N^{*}\) as maps \(\DR^{1}(Q)\to \Der(Q)\). We
denote by \(\pi^{N}\) the (unique) bivector in \(\mathcal{V}^{2}Q\) associated
to this map; then for every pair of 1-forms \(\alpha,\beta\in \DR^{1}(Q)\) we
have
\begin{equation}
  \label{eq:def-pi-N}
  \pi^{N}(\alpha,\beta) = \pi(N^{*}(\alpha),\beta) =
  \pi(\alpha,N^{*}(\beta)).
\end{equation}
Now let us introduce, again following \cite{ksm90}, two possible deformations
of the bracket \eqref{eq:def-br-1f}. The first one is simply its
\(N^{*}\)-deformed version, in the sense of definition \ref{def:nt}:
\begin{equation}
  \label{eq:def-b1f-1}
  \{\alpha,\beta\}_{\pi,N^{*}} = \{N^{*}(\alpha),\beta\}_{\pi} +
  \{\alpha,N^{*}(\beta)\}_{\pi} - N^{*}(\{\alpha,\beta\}_{\pi})\,.
\end{equation}
The second deformation is obtained by replacing the operators \(\de\) and
\(\ld{\theta}\) in the definition \eqref{eq:def-br-1f} with the operators
\(\de^{N}\) and  \(\ld{\theta}^{N}\) introduced above. The resulting bracket
reads
\begin{multline}
  \label{eq:def-b1f-2}
  \{\alpha,\beta\}_{\pi}' = 
  \ld{N(\tilde{\pi}(\alpha))}(\beta) - \ld{\tilde{\pi}(\alpha)}(N^{*}(\beta)) +
  N^{*}(\ld{\tilde{\pi}(\alpha)}(\beta)) +\\
  - \ld{N(\tilde{\pi}(\beta))}(\alpha) +
  \ld{\tilde{\pi}(\beta)}(N^{*}(\alpha)) -
  N^{*}(\ld{\tilde{\pi}(\beta)}(\alpha)) - N^{*}(\de(\pi(\alpha,\beta)))\,.
\end{multline}
We can now define the noncommutative version of the \emph{Magri-Morosi
  concomitant} as
\[ C_{(\pi,N)}(\alpha,\beta) \deq \frac{1}{2} \left(\{\alpha,\beta\}_{\pi,N^{*}} - \{\alpha,\beta\}_{\pi}'\right)\,. \]
By direct computation one sees that
\[ C_{(\pi,N)}(\alpha,\beta) = \ld{\tilde{\pi}(\alpha)}(N^{*}(\beta)) -
\ld{\tilde{\pi}(\beta)}(N^{*}(\alpha)) - \de (\pi^{N}(\alpha,\beta)) -
N^{*}\left( \ld{\tilde{\pi}(\alpha)}(\beta) - \ld{\tilde{\pi}(\beta)}(\alpha)
  -\de (\pi(\alpha,\beta))\right)\,. \]
We say that \(\pi\) and \(N\) are \emph{differentially compatible} if
\(C_{(\pi,N)}=0\); obviously this happens if and only if the two brackets
\eqref{eq:def-b1f-1} and \eqref{eq:def-b1f-2} coincide. As shown in
\cite{ksm90} this also implies that both these brackets coincide with the
bracket between 1-forms induced by the bivector \(\pi^{N}\) defined by
\eqref{eq:def-pi-N}.
\begin{defn}
  \label{def:compat}
  Let \(Q\) be a quiver, \(\pi\in \mathcal{V}^{2}Q\) a double Poisson
  structure and \(\app{N}{\Der(Q)}{\Der(Q)}\) a Nijenhuis tensor on
  \(\bb{k}Q\). We say that \(\pi\) and \(N\) are \emph{compatible} if
  \begin{enumerate}
  \item \(N\circ \tilde{\pi} = \tilde{\pi}\circ N^{*}\) and
  \item \(C_{(\pi,N)}=0\).
  \end{enumerate}
  In this case we call the pair consisting of \(\pi\) and \(N\) a
  \emph{Poisson-Nijenhuis structure on the path algebra} \(\bb{k}Q\).
\end{defn}
In this new setting the main result of the theory reads as follows:
\begin{teo}\label{mainth1}
  Let \(Q\) be a quiver and \((\pi, N)\) a Poisson-Nijenhuis structure on
  \(\bb{k}Q\). Then the bivector \(\pi^{N}\) defined by equation
  \eqref{eq:def-pi-N} is a double Poisson structure on \(\bb{k}Q\) which is
  compatible with \(\pi\).
\end{teo}
\begin{proof}
  As in the classical case \cite{ksm90} the result hinges on the following
  identity, which is valid for every double Poisson structure \(\pi\) and for
  every regular endomorphism \(N\) (not necessarily with zero torsion) which
  is algebraically compatible with it:
  \begin{equation}
    \label{eq:id-ksm}
    \pair{\mathcal{T}_{N^{*}}(\alpha,\beta)}{\theta} +
    \pair{\alpha}{\mathcal{T}_{N}(\tilde{P}(\beta),\theta)} =
    \pair{C_{(\pi,N)}(N^{*}(\alpha),\beta)}{\theta} -
    \pair{C_{(\pi,N)}(\alpha,\beta)}{N(\theta)}\,.
  \end{equation}
  This equality can be verified by a direct computation. In the hypotheses
  stated, identity \eqref{eq:id-ksm} implies that \(\mathcal{T}_{N^{*}}=0\).
  By theorem 1 it follows that the bracket \(\{\cdot,\cdot\}_{\pi,N^{*}}\)
  obeys the Jacobi identity and is compatible with \(\{\cdot,\cdot\}_{\pi}\).
  But the bracket \(\{\cdot,\cdot\}_{\pi,N^{*}}\) coincides with the bracket
  \(\{\cdot,\cdot\}_{\pi^{N}}\), so that \(\pi^{N}\) is a double Poisson
  structure by theorem 2 and is compatible with \(\pi\).
\end{proof}
Clearly the process can be iterated, so that on a quiver equipped with a PN
structure we have a whole hierarchy of double Poisson structures defined by
the maps \((\tilde{\pi}_{k})_{k\geq 0}\), where
\[ \tilde{\pi}_{k}\deq \underbrace{N\circ \dots\circ N}_{k\text{ times}}\circ \tilde{\pi}\,, \]
and every pair of such double Poisson structures is compatible.

An important special case is when the first Poisson structure is
\emph{invertible}, that is, when it comes from a noncommutative symplectic
structure on \(\bb{k}Q\). Let us recall \cite{kont93,ginz01} that a
\emph{noncommutative symplectic structure} on \(\bb{k}Q\) is given by a 2-form
\(\omega\in \DR^{2}(Q)\) which is closed (\(\de\omega=0\)) and non degenerate,
meaning that the map \(\app{\omega^{\flat}}{\Der(Q)}{\DR^{1}(Q)}\) defined by
\(\theta\mapsto i_{\theta}(\omega)\) is invertible. Let us denote by
\(\app{\omega^{\sharp}}{\DR^{1}(Q)}{\Der(Q)}\) its inverse; it maps a 1-form
\(\alpha\) to the unique derivation such that
\(i_{\omega^{\sharp}(\alpha)}(\omega) = \alpha\). For any \(f\in \DR^{0}(Q)\)
we have the corresponding ``Hamiltonian derivation'' \(\theta_{f} =
-\omega^{\sharp}(\de f)\).

The following lemma (already implicit in \cite{bie13}) clarifies the
relationship between symplectic forms and Poisson bivectors on quiver path
algebras.
\begin{lem}
  Suppose \(\omega\) is a non-degenerate 2-form on \(\bb{k}Q\) and let
  \(\pi\in \mathcal{V}^{2}Q\) be the bivector associated to
  \(-\omega^{\sharp}\) by the equality \eqref{eq:rel-pi-pit}. Then \(\omega\)
  is symplectic (\(\de\omega=0\)) if and only if \(\pi\) is a double Poisson
  structure (\([\pi,\pi]=0\)).
\end{lem}
\begin{proof}
  The 2-form \(\omega\) defines a bilinear, skew-symmetric bracket on
  \(\DR^{0}(Q)\) by the usual prescription:
  \[ \{f,g\}\deq i_{\theta_{g}}(i_{\theta_{f}}(\omega))\,. \]
  This induces a bilinear and skew-symmetric bracket on the space of
  \(G_{\mathbf{n}}\)-invariant functions defined on every representation space
  \(\Rep_{\bb{k}}(Q,\mathbf{n})\). By known results, the latter brackets obey the
  Jacobi identity if and only if the induced 2-forms \(\hat{\omega}\) are
  closed, and this happens if and only if \(\de\omega=0\).

  On the other hand, the bracket defined on \(\DR^{0}(Q)\) by the noncommutative
  bivector \(\pi\) is the same as above, since
  \[ \pi(\de f, \de g) = \pair{\de g}{-\omega^{\sharp}(\de f)} =
  i_{\theta_{f}}(\de g) = i_{\theta_{f}}(-\omega^{\flat}(\theta_{g})) =
  -i_{\theta_{f}}(i_{\theta_{g}}(\omega)) = i_{\theta_{g}}(i_{\theta_{f}}(\omega))\,. \]
  It follows that the induced brackets on representation spaces obey the
  Jacobi identity if and only if the induced bivectors \(\check{\pi}\) are
  Poisson. By theorem \ref{teo:bie} this happens if and only if
  \([\pi,\pi]=0\).
\end{proof}
Hence every noncommutative symplectic structure on \(\bb{k}Q\) gives rise to a
unique double Poisson structure on \(\bb{K}Q\), exactly as in the classical
setting. We also have the following analogue of another well-known result.
\begin{teo}
  \label{compat-sympl}
  Suppose \(\omega\) is a noncommutative symplectic form on \(\bb{k}Q\) with
  associated double Poisson structure \(\pi_{0}\) and let \(\pi_{1}\) be
  another double Poisson structure on \(\bb{k}Q\). Then:
  \begin{enumerate}
  \item the endomorphism \(\app{N}{\Der(Q)}{\Der(Q)}\) defined by
    \(\tilde{\pi}_{1}\circ (-\omega^{\flat})\) is regular;
  \item if \(\pi_{0}\) and \(\pi_{1}\) are compatible then \(N\) is Nijenhuis
    and compatible with \(\pi_{0}\).
  \end{enumerate}
\end{teo}
\begin{proof}
  (1) It suffices to show that there exists a map \(\app{N^{*}}{\DR^{1}(Q)}
  {\DR^{1}(Q)}\) such that \(\pair{N^{*}(\alpha)}{\theta}\) equals
  \[ \pair{\alpha}{N(\theta)} =
  -\pair{\alpha}{\tilde{\pi}_{1}(\omega^{\flat}(\theta))} =
  -\pair{\alpha}{\tilde{\pi}_{1}(i_{\theta}(\omega))} =
  -\pi_{1}(i_{\theta}(\omega),\alpha) = \pi_{1}(\alpha,i_{\theta}(\omega)) \]
  for any \(\alpha\in \DR^{1}(Q)\), \(\theta\in \Der(Q)\). We claim that
  \(N^{*}\deq -\omega^{\flat}\circ \tilde{\pi}_{1}\) does the job: indeed,
  \[ \pair{-\omega^{\flat}(\tilde{\pi}_{1}(\alpha))}{\theta} =
  -i_{\theta}(i_{\tilde{\pi}_{1}(\alpha)}(\omega)) =
  i_{\tilde{\pi}_{1}(\alpha)}(i_{\theta}(\omega)) =
  \pair{i_{\theta}(\omega)}{\tilde{\pi}_{1}(\alpha)} =
  \pi_{1}(\alpha,i_{\theta}(\omega))\,, \]
  as we wanted.

  (2) The first assertion follows from the following identity which, in the
  hypotheses stated, relates the torsion of \(N\) to the Schouten bracket of
  \(\pi_{0}\) and \(\pi_{1}\):
  \[ \mathcal{T}_{N}(\theta,\eta) =
  2N([\pi_{0},\pi_{1}](\omega^{\flat}(\theta),\omega^{\flat}(\eta)))\,. \]
  By a direct computation one then shows that \(C_{(\pi,N)}=0\).
\end{proof}
As is customary, a PN structure in which one of the two Poisson bivectors
comes from a symplectic form will be called a \(\omega N\) structure.

\subsection{Noncommutative lifts}
\label{lifts}

Now we would like to reinterpret in our setting the construction of compatible
Poisson brackets on cotangent bundles introduced in \cite{tur92} and used
in  \cite{bfmop10} to obtain the bihamiltonian structure of the Calogero-Moser system.

The noncommutative analogue of cotangent bundles are double quivers, so let us
consider a quiver \(Q\) and its double \(\ol{Q}\), where for each arrow \(a\)
we denote its opposite by \(a^{*}\). We have the tautological 1-form
\(\lambda\in \DR^{1}(\ol{Q})\) represented by the expression \(\sum_{a\in Q}
a^{*}\de a\) and the corresponding canonical symplectic form \(\omega =
\de\lambda\) in \(\DR^{2}(\ol{Q})\), represented by
\[ \omega = \sum_{a\in Q} \de a^{*}\de a \,.\]
Consider now a regular endomorphism \(\app{L}{\Der(Q)}{\Der(Q)}\) on the path
algebra of \(Q\). We define the following deformation of \(\lambda\),
\[ \lambda_{L}\deq \sum_{a\in Q} a^{*}L^{*}(\de a) = \sum_{a\in Q} a^{*}\de^{L}a\,, \]
and denote by \(\omega_{L}\) its differential (which is not a symplectic form
in general).

The \emph{complete lift} of \(L\) to the double \(\ol{Q}\) is the map
\(\app{N}{\Der(\ol{Q})}{\Der(\ol{Q})}\) defined by
\begin{equation}
  \label{eq:def-lift}
  \theta\mapsto \omega^{\sharp}(i_{\theta}(\omega_{L}))\,.
\end{equation}
In other words, \(N(\theta)\) is the unique derivation of \(\bb{k}\ol{Q}\)
such that
\[ i_{N(\theta)}(\omega) = i_{\theta}(\omega_{L})\,. \]
\begin{lem}
  \(N\) is a regular endomorphism of \(\Der(\ol{Q})\).
\end{lem}
\begin{proof}
  It suffices to show that \(N\) has a transpose. Let us write the generic
  1-form in \(\DR^{1}(\ol{Q})\) as \(\omega^{\flat}(\eta)\), with \(\eta\in
  \Der(\ol{Q})\). We need a map \(\app{N^{*}}{\DR^{1}(\ol{Q})}{\DR^{1}(\ol{Q})}\) such that
  \(\pair{N^{*}(\omega^{\flat}(\eta))}{\theta}\) equals
  \[ \pair{\omega^{\flat}(\eta)}{N(\theta)} = i_{N(\theta)}(i_{\eta}(\omega))
  = -i_{\eta}(i_{N(\theta)}(\omega)) = -i_{\eta}(i_{\theta}(\omega_{L})) =
  i_{\theta}(i_{\eta}(\omega_{L})) \]
  for every \(\theta\in \Der(\ol{Q})\). Clearly then we should take
  \[ N^{*}(\omega^{\flat}(\eta)) = i_{\eta}(\omega_{L})\,, \]
  that is, \(N^{*}(\beta)\deq i_{\omega^{\sharp}(\beta)}(\omega_{L})\) for
  every \(\beta\in \DR^{1}(\ol{Q})\).
\end{proof}
Let us consider now the bivector defined by the map
\[ \tilde{\pi}_{1}\deq -N\circ \omega^{\sharp}\,. \]
Explicitly, one has
\[ \pi_{1}(\alpha,\beta) = i_{\omega^{\sharp}(\beta)}(i_{\omega^{\sharp}(\alpha)}(\omega_{L}))\,, \]
as may be verified by a direct computation. This is the noncommutative version
of the bivector considered in \cite{tur92,imm00}.
\begin{teo}
  \label{teo:lifts}
  If \(\mathcal{T}_{L}=0\) then the bivector \(\pi_{1}\) is Poisson and
  compatible with the canonical Poisson structure.
\end{teo}
\begin{proof}
  The noncommutative bivector \(\pi_{1}\) induces a genuine bivector
  \(\check{\pi}_{1}\) on every representation space
  \(\Rep_{\bb{k}}(Q,\mathbf{n})\), and consequently a bracket on
  \(G_{\mathbf{n}}\)-invariant regular functions. The proofs in \cite{tur92}
  and \cite{imm00} then show that these brackets obey the Jacobi identity when
  \(L\) is torsionless. This means that \(\pi_{1}\) induces a Poisson bivector
  on every representation space, hence theorem \ref{teo:bie} implies that
  \(\pi_{1}\) is a double Poisson structure on \(\bb{k}\ol{Q}\). Compatibility
  with \(\pi_{0}\) then follows from the corresponding property of induced
  bivectors.
\end{proof}
We conclude that a Nijenhuis tensor \(L\) on \(\bb{k}Q\) is enough to induce a
\(\omega N\) structure on the path algebra \(\bb{k}\ol{Q}\). It is important
to emphasize that the Nijenhuis tensors obtained by this lifting process are
quite special: for example the action of \(N(\theta)\) on an arrow \(a\in Q\)
cannot involve any of the arrows in \(\ol{Q}\setminus Q\). As we shall see in
the next section, this is often a serious limitation.

\section{Examples and applications}
\label{sec:ex}

\subsection{Rational Calogero-Moser system}
\label{CM}

As a first example let us consider the noncommutative \(\omega N\) manifold that
underlies the phase spaces of the rational Calogero-Moser systems.

Let \(Q_{\circ}\) be the quiver with one vertex and one loop \(a\) and denote
by \(\ol{Q}_{\circ}\) its double (which has an additional loop \(a^{*}\)). The
corresponding path algebra is the free associative algebra on the two
generators \(a\) and \(a^{*}\):
\[ \bb{k}\ol{Q}_{\circ} = \falg{\bb{k}}{a,a^{*}}\,. \]
The tautological 1-form on this path algebra is \(\lambda = a^{*}\de a\) and
the canonical symplectic form reads
\begin{equation}
  \label{eq:symp-cm}
  \omega = \de a^{*}\, \de a\,.
\end{equation}
The map \(\app{\omega^{\sharp}}{\DR^{1}(\ol{Q}_{\circ})}{\Der(\ol{Q}_{\circ})}\)
acts as follows: given a 1-form \(\alpha\in \DR^{1}(\ol{Q}_{\circ})\)
represented by \(\sum_{c\in \ol{Q}_{\circ}} S_{c}\de c\),
\begin{equation}
  \label{eq:om-sh-cm}
  \omega^{\sharp}(\alpha) = -S_{a^{*}}\partial_{a} + S_{a}\partial_{a^{*}}
\end{equation}
so that the Poisson bivector associated to \(\omega\) is simply
\[ \pi_{0} = [\partial_{a^{*}},\partial_{a}]\,. \]
Following \cite{bfmop10}, let us consider the endomorphism
\(\app{L}{\Der(Q_{\circ})}{\Der(Q_{\circ})}\) defined as follows: for every
\(\theta\in \Der(Q_{\circ})\), \(L(\theta)\) is the unique derivation of
\(\bb{k}Q_{\circ}\) mapping \(a\) to \(a\theta(a)\). It is straightforward to
verify that \(L\) is regular (its transpose being given by \(L^{*}(\de a) =
a\de a\)) and a Nijenhuis tensor on \(\bb{k}Q_{\circ}\). The corresponding
deformed tautological 1-form on \(\bb{k}\ol{Q}_{\circ}\) is
\begin{equation}
  \label{eq:taut-cm}
  \lambda_{L} = a^{*}a\de a
\end{equation}
with differential
\[ \omega_{L} = \de\lambda_{L} = \de a^{*}\, a\de a + a^{*}\de a\de a\,. \]
The complete lift of \(L\), as defined by equation \eqref{eq:def-lift}, is
then obtained as follows. First we notice that
\begin{align*}
  i_{\theta}(\omega_{L})) &= \theta(a^{*})a\de a - \de a^{*}\, a\theta(a) +
  a^{*}\theta(a)\de a - a^{*}\de a\, \theta(a)\\
  &= (\theta(a^{*})a + a^{*}\theta(a) - \theta(a)a^{*})\de a - a\theta(a) \de
  a^{*}\,,
\end{align*}
where the second equality holds in \(\DR^{1}(\ol{Q}_{\circ})\). Then, using
the expression \eqref{eq:om-sh-cm} for \(\omega^{\sharp}\), we get
\[ N(\theta) = \omega^{\sharp}(i_{\theta}(\omega_{L})) =
a\theta(a)\partial_{a} + (\theta(a^{*})a + a^{*}\theta(a) -
\theta(a)a^{*})\partial_{a^{*}} \]
or, more compactly,
\begin{equation}
  \label{eq:Ncm}
  N(\theta)(a,a^{*}) = (a\theta(a), [a^{*},\theta(a)] + \theta(a^{*})a)\,.
\end{equation}
Theorem \ref{teo:lifts} then implies that the map
\[ \tilde{\pi}_{1}(S_{a}\de a + S_{a^{*}}\de a^{*}) = N(S_{a^{*}}\partial_{a} - S_{a}\partial_{a^{*}}) =
aS_{a^{*}}\partial_{a} + (a^{*}S_{a^{*}} - S_{a^{*}}a^{*} - S_{a}a)\partial_{a^{*}} \]
defines a double Poisson structure on \(\bb{k}\ol{Q}_{\circ}\). Explicitly,
the corresponding bivector \(\pi_{1}\in \mathcal{V}^{2}\ol{Q}_{\circ}\) reads
\begin{equation}
  \label{eq:pi1-cm}
  \pi_{1} = [a\partial_{a^{*}},\partial_{a}] +
  [a^{*}\partial_{a^{*}},\partial_{a^{*}}]\,.
\end{equation}
Readers of \cite{bie13,ors13} will recognize the previous expression as the
linear Poisson bivector on \(\falg{\bb{k}}{a,a^{*}}\) induced by an
appropriate associative algebra structure on \(\bb{k}^{2}\).

Let us compute the next double Poisson structure in the hierarchy. We have
\begin{align*}
  \tilde{\pi}_{2}(S_{a}\de a + S_{a^{*}}\de a^{*}) &= N(\tilde{\pi}_{1}(S_{a}\de a + S_{a^{*}}\de a^{*}))\\
  &= N(aS_{a^{*}}\partial_{a} + ([a^{*},S_{a^{*}}] - S_{a}a)\partial_{a^{*}})\\
  &= a^{2}S_{a^{*}}\partial_{a} + ([a^{*},aS_{a^{*}}] + [a^{*},S_{a^{*}}]a - S_{a}a^{2})\partial_{a^{*}}\,,
\end{align*}
which corresponds to the bivector
\begin{equation}
  \label{eq:p2-cm}
  \pi_{2} = [a^{2}\partial_{a^{*}},\partial_{a}] + [a^{*}a\partial_{a^{*}},\partial_{a^{*}}]
  + [a^{*}\partial_{a^{*}},a\partial_{a^{*}}]\,. 
\end{equation}
In general, we have
\[ \tilde{\pi}_{m}(S_{a}\de a + S_{a^{*}}\de a^{*}) = a^{m}S_{a^{*}}\partial_{a} +
\left(\sum_{i=1}^{m} [a^{*},a^{m-i}S_{a^{*}}]a^{i-1} - S_{a}a^{m}\right)\partial_{a^{*}} \]
whence
\[ \pi_{m} = [a^{m}\partial_{a^{*}},\partial_{a}] + \sum_{i=1}^{m}
[a^{*}a^{m-i}\partial_{a^{*}}, a^{i-1}\partial_{a^{*}}]\,. \]
In particular the Poisson brackets on \(\DR^{0}(\ol{Q}_{\circ})\) determined
by the \(m\)-th Poisson structure read as follows:
\begin{equation}
  \label{eq:bracket-m}
  \{f,g\}_{m} = a^{m} \left( \frac{\partial f}{\partial a^{*}} 
    \frac{\partial g}{\partial a} - \frac{\partial g}{\partial a^{*}}
    \frac{\partial f}{\partial a}\right) + \sum_{i=1}^{m} a^{*} a^{m-i}
  \left( \frac{\partial f}{\partial a^{*}} a^{i-1}\frac{\partial g}{\partial a^{*}}
    - \frac{\partial g}{\partial a^{*}} a^{i-1}\frac{\partial f}{\partial a^{*}}\right)
\end{equation}
where \(\frac{\partial}{\partial a}\) denotes the \emph{necklace derivative}
with respect to the arrow \(a\) \cite{kont93,ginz01,blb02}.

Consider now the following family of necklace words in \(\DR^{0}(\ol{Q}_{\circ})\):
\begin{equation}
  \label{eq:def-Ik}
  I_{k} = \frac{1}{k}a^{k} \qquad (k\geq 1)\,.
\end{equation}
As is immediate to verify, these regular functions on \(\bb{k}\ol{Q}_{\circ}\)
are in involution with respect to every bracket of the hierarchy
\eqref{eq:bracket-m}. Moreover, one has
\[ \tilde{\pi}_{1}(\de I_{k}) = \tilde{\pi}_{1}(a^{k-1}\de a) =
-a^{k-1}a\partial_{a^{*}} = -a^{k}\partial_{a^{*}} \]
and
\[ \tilde{\pi}_{0}(\de I_{k+1}) = \tilde{\pi}_{0}(a^{k}\de a) = -a^{k}\partial_{a^{*}} \]
so that the functions \(I_{k}\) form a \emph{Lenard chain}, that is
\[ \tilde{\pi}_{1}(\de I_{k}) = \tilde{\pi}_{0}(\de I_{k+1})\,. \]
For later use, let us define also the following additional set of regular
functions on \(\bb{k}\ol{Q}_{\circ}\):
\begin{equation}
  \label{eq:def-Jl}
  J_{\ell}\deq a^{\ell-1}a^{*} \qquad (\ell\geq 1)\,.
\end{equation}
Taken together, the \(I_{k}\) and \(J_{\ell}\) span a Lie subalgebra of
\(\DR^{0}(\ol{Q}_{\circ})\) with respect to each one of the brackets
\eqref{eq:bracket-m}. Indeed, the following relations hold for every
\(k,\ell\geq 1\) and \(m\geq 0\):
\begin{equation}
  \begin{array}{ccc}
    \{I_{k},I_{\ell}\}_{m} = 0 & \{J_{\ell},I_{k}\}_{m} = (k+\ell+m-2)I_{k+\ell+m-2} &
    \{J_{k},J_{\ell}\}_{m} = (\ell-k)J_{k+\ell+m-2}
  \end{array}
\end{equation}
(with the exception that \(\{J_{1},I_{1}\}_{0} = 1\)). 

In order to relate the above constructions with the dynamics of the rational
Calogero-Moser system let us descend to the space of real representations of
the quiver \(\ol{Q}_{\circ}\) with dimension vector \(\mathbf{n} = (n)\) for
some \(n\in \bb{N}\). This is simply the linear space of pairs of \(n\times
n\) real matrices,
\begin{equation}
  \label{eq:rep-cm}
  \Rep_{\bb{R}}(\ol{Q}_{\circ},(n)) = \Mat_{n\times n}(\bb{R})\oplus
  \Mat_{n\times n}(\bb{R})\,,
\end{equation}
which can be identified with the cotangent bundle \(T^{*}\Mat_{n\times n}
(\bb{R})\) in the obvious way. The group \(G_{\mathbf{n}}\) defined by
\eqref{eq:group} coincides with \(\PGL_{n}(\bb{R})\) and its action on the
space \eqref{eq:rep-cm} is Hamiltonian with respect to the (canonical)
symplectic form \(\hat{\omega}\) on \(T^{*}\Mat_{n\times n}(\bb{R})\) induced
by the noncommutative symplectic form \eqref{eq:symp-cm}.
As well known (see e.g. \cite[Chapter 2]{etin06}), the phase space of the
rational \(n\)-particle Calogero-Moser system may then be recovered as a
suitable symplectic quotient of the manifold \((T^{*}\Mat_{n\times n}
(\bb{R}),\hat{\omega})\), and the \(\PGL_{n}(\bb{R})\)-invariant functions
on \(\Rep_{\bb{R}}(\ol{Q}_{\circ},(n))\) induced by the noncommutative
functions \eqref{eq:def-Ik},
\[ \hat{I}_{k}(X,Y) = \frac{1}{k}\tr X^{k}\,, \]
give exactly the usual Calogero-Moser Hamiltonians once projected on this
quotient space. (Of course, on the resulting finite-dimensional manifold only
the first \(n\) of them will be functionally independent.) 

This procedure is not appropriate in the present setting, however, as the
quotient map obtained by the symplectic reduction process cannot be used to
reduce the second Poisson structure \(\check{\pi}_{1}\) on
\(\Rep_{\bb{R}}(\ol{Q}_{\circ},(n))\). To do this we need to replace the ordinary
symplectic reduction with a ``bihamiltonian reduction''.

As described in detail in \cite{bfmop10}, a reduction of this kind is
naturally viewed as a two-step process. In the first step one factors out the
action of \(\PGL_{n}(\bb{R})\) on \(\Rep_{\bb{R}}(\ol{Q}_{\circ},(n))\),
landing on a certain \((n^{2}+1)\)-dimensional manifold \(\mathcal{P}\). In
the second step one further reduces the resulting dynamics on the
\(2n\)-dimensional manifold defined by the image of the submersion
\(\app{\pi}{\mathcal{P}}{\bb{R}^{2n}}\) whose components are the invariant
functions \(\hat{I}_{1}, \dots, \hat{I}_{n}\) and \(\hat{J}_{1}, \dots,
\hat{J}_{n}\) (where \(\hat{J}_{\ell}(X,Y) = \tr X^{\ell-1}Y\)). This makes it
possible to recover, for any fixed \(n\), the phase space of the (attractive)
\(n\)-particles rational Calogero-Moser system with its associated
bihamiltonian structure.
\begin{rem}
  We believe that a similar quotient can be constructed also starting from the
  space of \emph{complex} representation of the quiver \(\ol{Q}_{\circ}\).
  This generalization is needed in order to get the dynamics of the repulsive
  Calogero-Moser system.
\end{rem}
\begin{rem}
  Unfortunately it is not easy to write the second Poisson structure in the
  reduced (or ``physical'') coordinates. As a matter of fact, the
  transformation relating the functions \((\hat{I}_{1}, \dots, \hat{I}_{n},
  \hat{J}_{1}, \dots, \hat{J}_{n})\) to the canonical Calogero-Moser
  coordinates \((q_{1}, \dots, q_{n}, p_{1}, \dots, p_{n})\) is notoriously
  hard to invert. In \cite{mm96} and \cite{bfmop10} only the 2-particle case
  is considered; more recently the 3-particle case has been studied in
  \cite{ar12}.
\end{rem}
\begin{rem}
  In the literature concerning the construction of the Calogero-Moser phase
  space by symplectic reduction it is more customary to take as Hamiltonians
  the functions \(\hat{H}_{k}(X,Y) = \frac{1}{k}\tr Y^{k}\) and
  \(\hat{K}_{\ell}(X,Y) = \tr Y^{\ell-1}X\), which are induced by the
  following necklace words in \(\DR^{0}(\ol{Q}_{\circ})\):
  \[ H_{k} = \frac{1}{k}a^{*k} \quad\text{ and }\quad K_{\ell} = a^{*\ell-1}a\,. \]
  These functions also define a bihamiltonian system on
  \(\bb{k}\ol{Q}_{\circ}\) (and consequently on each representation space) if
  we replace the Nijenhuis tensor \eqref{eq:Ncm} with
  \begin{equation}
    \label{eq:alt-Ncm}
    N(\theta)(a,a^{*}) = ([\theta(a^{*}),a] + a^{*}\theta(a), \theta(a^{*})a^{*})\,,
  \end{equation}
  in which case the second Poisson structure turns out to be
  \[ \pi_{1} = [a^{*}\partial_{a^{*}},\partial_{a}] +
  [a\partial_{a},\partial_{a}]\,. \]
  In this paper we stuck to the choice \eqref{eq:Ncm} in order to ease the
  comparison between our formulas and the corresponding ones in reference
  \cite{bfmop10}. Notice in this respect that the Nijenhuis tensor
  \eqref{eq:alt-Ncm} cannot be obtained by a lifting process of the kind
  discussed in subsection \ref{lifts}.
\end{rem}

\subsection{Gibbons-Hermsen system}
\label{GH}

As a second example we shall consider a noncommutative \(\omega N\) manifold
related to a family of integrable systems introduced by Gibbons and Hermsen in
\cite{gh84}. These systems are a generalization of the rational Calogero-Moser
system in which each particle has some additional degrees of freedom
parametrized by a vector-covector pair living in a linear space of dimension
\(r>1\) (the case \(r=1\) corresponds to Calogero-Moser). For the sake of
notational clarity we shall consider only the case \(r=2\); however, the
generalization to higher-rank cases does not present any essentially new
difficulty.

Let \(Q\) be the quiver
\begin{equation}
  \label{ghquiver}
  \xymatrix{
    \bullet \ar@(ul,dl)[]_{a} \ar@/_0.5pc/[r]_{y} & \bullet \ar@/_0.5pc/[l]_{x}}
\end{equation}
introduced by Bielawski and Pidstrygach in \cite{bp08}. On the path algebra of
its double \(\ol{Q}\) we have the tautological 1-form
\begin{equation}
  \label{eq:taut-qbp}
  \lambda = a^{*}\de a + x^{*}\de x + y^{*} \de y\,.
\end{equation}
The corresponding symplectic form is
\[ \omega = \de a^{*}\de a + \de x^{*}\de x + \de y^{*}\de y\,. \]
The associated map \(\app{\omega^{\sharp}}{\DR^{1}(\ol{Q})}{\Der(\ol{Q})}\)
acts on a generic 1-form \(\alpha = \sum_{c\in \ol{Q}} S_{c}\de c\) in the
following manner:
\begin{equation}
  \label{eq:om-sh-gh}
  \omega^{\sharp}(\alpha) = -S_{a^{*}}\partial_{a} + S_{a}\partial_{a^{*}}
  - S_{x^{*}}\partial_{x} + S_{x}\partial_{x^{*}}
  - S_{y^{*}}\partial_{y} + S_{y}\partial_{y^{*}}\,.
\end{equation}
This symplectic form will provide our first Poisson bivector on \(\bb{k}\ol{Q}\),
\[ \pi_{0} = [\partial_{a^{*}},\partial_{a}] + [\partial_{x^{*}},\partial_{x}] +
[\partial_{y^{*}},\partial_{y}]\,. \]
Now let us consider the 1-form
\[ \lambda' = a^{*}a\de a + x^{*}a\de x - ya\de y^{*}\,. \]
Notice that \(\lambda'\) cannot be expressed as a deformation of the 1-form
\eqref{eq:taut-qbp} via a regular endomorphism of \(\Der(Q)\) because of the
term involving \(\de y^{*}\). It is, however, a rather natural extension of
the 1-form \eqref{eq:taut-cm} to the new setting.

Let us continue anyway along the same track by defining an endomorphism
\(\app{N}{\Der(\ol{Q})}{\Der(\ol{Q})}\) using equation \eqref{eq:def-lift},
where the role of \(\omega_{L}\) is now played by the 2-form
\[ \de\lambda' = \de a^{*}\, a\de a + a^{*}\de a\de a + \de x^{*}\, a\de x +
x^{*}\de a\de x - \de y\, a\de y^{*} - y\de a\de y^{*}\,. \]
By contracting with a generic derivation \(\theta\) we get
\begin{multline*}
  i_{\theta}(\de\lambda') = 
  \left(\theta(a^{*}) a + [a^{*},\theta(a)] - \theta(x)x^{*} + \theta(y^{*})y\right)\de a
  - a\theta(a)\de a^{*} + \\
  + \left(\theta(x^{*})a + x^{*}\theta(a)\right)\de x
  - a\theta(x)\de x^{*}
  + a\theta(y^{*})\de y
  - \left(\theta(y)a + y\theta(a)\right)\de y^{*}
\end{multline*}
so that, using \eqref{eq:om-sh-gh}
\begin{multline}
  \label{eq:Ngh}
  N(\theta) = a\theta(a)\partial_{a} + \left(\theta(a^{*}) a +
    [a^{*},\theta(a)] - \theta(x)x^{*} + \theta(y^{*})y\right)\partial_{a^{*}} +\\
  + a\theta(x)\partial_{x} + \left(\theta(x^{*})a + x^{*}\theta(a)\right)\partial_{x^{*}}
  + \left(\theta(y)a + y\theta(a)\right)\partial_{y} + a\theta(y^{*})\partial_{y^{*}}\,.
\end{multline}
Now let us consider the map \(\tilde{\pi}_{1}\deq N\circ \tilde{\pi}_{0}\).
Recalling that \(\tilde{\pi}_{0} = -\omega^{\sharp}\), we get
\begin{align*}
  \tilde{\pi}_{1}(\alpha) &= N(S_{a^{*}}\partial_{a} - S_{a}\partial_{a^{*}}
  + S_{x^{*}}\partial_{x} - S_{x}\partial_{x^{*}}
  + S_{y^{*}}\partial_{y} - S_{y}\partial_{y^{*}}) =\\
  &= aS_{a^{*}}\partial_{a} + (-S_{a}a + [a^{*},S_{a^{*}}] - S_{x^{*}}x^{*} -
  S_{y}y)\partial_{a^{*}} + \\
  &\phantom{=} + aS_{x^{*}}\partial_{x} + (-S_{x}a + x^{*}S_{a^{*}})\partial_{x^{*}}
    + (S_{y^{*}}a + yS_{a^{*}})\partial_{y} - aS_{y}\partial_{y^{*}} \,.
\end{align*}
The corresponding bivector in \(\mathcal{V}^{2}\ol{Q}\) is given by
\[ \pi_{1} = [a\partial_{a^{*}},\partial_{a}] + [a^{*}\partial_{a^{*}},\partial_{a^{*}}]
+ [a\partial_{x^{*}},\partial_{x}] + [x^{*}\partial_{a^{*}},\partial_{x^{*}}]
+ [\partial_{y^{*}},a\partial_{y}] + [y\partial_{a^{*}},\partial_{y}]\,. \]
A straightforward computation using definition \ref{schouten} reveals that
\[ [\pi_{1},\pi_{1}] = 0 \quad\text{ and }\quad [\pi_{0},\pi_{1}] = 0\,, \]
i.e. that \(\pi_{1}\) is a double Poisson structure on \(\ol{Q}\) which is
compatible with \(\pi_{0}\). A posteriori we can conclude, using theorem
\ref{compat-sympl}, that the endomorphism \(N\) defined by \eqref{eq:Ngh} is
Nijenhuis and compatible with \(\pi_{0}\) (in the sense that \(C_{(\pi_{0},N)}=0\)).

Let us take as Hamiltonians the necklace words in \(\DR^{0}(\ol{Q})\) of the
following form:
\begin{subequations}
  \label{eq:ham-gh-nc}
  \begin{equation}
    I_{k}\deq \frac{1}{k} a^{k} \quad (k\geq 1)
  \end{equation}
  and
  \begin{equation}
    I^{(2)}_{k}\deq a^{k}(xx^{*} + y^{*}y) \quad (k\geq 0)\,.
  \end{equation}
\end{subequations}
It is clear that the functions \(I_{k}\) are linked in a Lenard chain by the
two Poisson structures \(\pi_{0}\) and \(\pi_{1}\); the computation is
essentially the same as in the Calogero-Moser case. Here, however, the
additional Hamiltonians \(I^{(2)}_{k}\) determine another Lenard chain: in
fact we have
\[ \de I^{(2)}_{k} = \sum_{i=0}^{k-1} a^{i}(xx^{*} + y^{*}y)a^{k-1-i} \de a +
x^{*}a^{k} \de x + a^{k}x\de x^{*} + a^{k}y^{*} \de y + ya^{k}\de y^{*} \]
so that
\[ \tilde{\pi}_{1}(\de I^{(2)}_{k}) = (-\sum_{i=0}^{k-1} a^{i}(xx^{*} + y^{*}y)
a^{k-i} - a^{k}xx^{*} - a^{k}y^{*}y)\partial_{a^{*}} + a^{k+1}x\partial_{x} -
x^{*}a^{k+1}\partial_{x^{*}} + ya^{k+1}\partial_{y} - a^{k+1}y^{*}\partial_{y^{*}}\,, \]
which equals
\[ \tilde{\pi}_{0}(\de I^{(2)}_{k+1}) = - \sum_{i=0}^{k} a^{i}(xx^{*} +
y^{*}y)a^{k-i}\partial_{a^{*}} + a^{k+1}x\partial_{x} -
x^{*}a^{k+1}\partial_{x^{*}} + ya^{k+1}\partial_{y} - a^{k+1}y^{*}\partial_{y^{*}}\,. \]
We conclude that the noncommutative functions \eqref{eq:ham-gh-nc} induce a
bihamiltonian system on every representation space for the quiver \(\ol{Q}\).

To explain the relationship with the Gibbons-Hermsen system let us consider
the space of real representations of \(\ol{Q}\) with dimension vector
\(\mathbf{n} = (n,1)\),
\[ \Rep_{\bb{R}}(\ol{Q},(n,1)) = 
\Mat_{n\times n}(\bb{R})\oplus \Mat_{n\times n}(\bb{R})\oplus 
\Mat_{n\times 1}(\bb{R})\oplus \Mat_{n\times 1}(\bb{R})\oplus
\Mat_{1\times n}(\bb{R})\oplus \Mat_{1\times n}(\bb{R})\,, \]
the point corresponding to a representation \(\tau\) being given by the
matrices \((\tau_{a}, \tau_{a^{*}}, \tau_{x}, \tau_{y^{*}}, \tau_{x^{*}},
\tau_{y})\). As explained in \cite{bp08}, the phase space of the rank \(2\)
Gibbons-Hermsen system can be obtained by identifying
\(\Rep_{\bb{R}}(\ol{Q},(n,1))\) with the linear space 
\[ V_{n,2}\deq \Mat_{n\times n}(\bb{R})\oplus \Mat_{n\times n}(\bb{R})\oplus
\Mat_{n\times 2}(\bb{R})\oplus \Mat_{2\times n}(\bb{R}) \]
consisting of quadruples \((X,Y,v,w)\) using the bijective correspondence
defined as follows:
\begin{equation}
\label{eq:iso-gh}
X = \tau_{a} \qquad
Y = \tau_{a^{*}} \qquad
v =
\begin{pmatrix}
  -\tau_{x} & \tau_{y^{*}}
\end{pmatrix} \qquad
w =
\begin{pmatrix}
  \tau_{x^{*}}\\
  \tau_{y}
\end{pmatrix}\,.
\end{equation}
In this way the natural Hamiltonian action of the group \(G_{(n,1)}\simeq
\GL_{n}(\bb{R})\) on the symplectic manifold \((\Rep_{\bb{R}}(\ol{Q},(n,1)),
\hat{\omega})\) coincides with the Hamiltonian action of \(\GL_{n}(\bb{R})\)
on \(V_{n,2}\) used in \cite{gh84} to define the phase space of the system by
symplectic reduction.

The dynamics of the system is determined by taking as Hamiltonians the
functions\footnote{As in the previous subsection we modify the usual
  Hamiltonians by exchanging the matrices \(X\) and \(Y\).}
\begin{equation}
  \label{eq:ham-gh}
  \hat{I}_{k}(X,Y,v,w) = \frac{1}{k}\tr X^{k} \quad\text{ and }\quad
  \hat{H}_{k,\alpha}(X,Y,v,w) = \tr X^{k}v\alpha w\,,
\end{equation}
where \(\alpha\) is any \(2\times 2\) constant matrix (actually
\(\hat{I}_{k}\) is just a scalar multiple of \(\hat{H}_{k,\iota}\), where
\(\iota\) is the identity matrix). These functions span a (nonabelian) Poisson
algebra \(\mathcal{H}\) whose Poisson brackets are given by
\[ \{\hat{H}_{k,\alpha}, \hat{H}_{\ell,\beta}\} = \hat{H}_{k+\ell,[\alpha,\beta]}\,. \]
The complete integrability of the system then follows from the existence of
\(2n\)-dimensional abelian subalgebras of \(\mathcal{H}\). A natural choice is
to take the subalgebra spanned by the functions \((\hat{I}_{1}, \dots,
\hat{I}_{n})\) and \((\hat{I}^{(2)}_{0}, \dots, \hat{I}^{(2)}_{n-1})\) where
\[ \hat{I}^{(2)}_{k}\deq \hat{H}_{k,\eta}\,, \qquad
\eta =
\begin{pmatrix}
  -1 & 0\\
  0 & 1
\end{pmatrix}.
\]
Using the correspondence \eqref{eq:iso-gh} it is immediate to check that these
functions are induced, respectively, by the necklace words \(I_{k}\) and
\(I^{(2)}_{k}\) in \(\DR^{0}(\ol{Q})\) defined by \eqref{eq:ham-gh-nc}. It
follows that the dynamics described by these functions on \(V_{n,2}\simeq
\Rep_{\bb{R}}(\ol{Q},(n,1))\) is bihamiltonian with respect to the induced
Poisson structures \(\check{\pi}_{0}\) and \(\check{\pi}_{1}\).

Unfortunately this is not enough to conclude that the induced dynamics on the
quotient manifold is also bihamiltonian. In fact we are faced with exactly the
same problem that arose in the Calogero-Moser case: in order to reduce the
second Poisson structure we cannot use the projection map coming from the
symplectic reduction \emph{\`a\ la} Gibbons-Hermsen, as this map will not
preserve \(\pi_{1}\). Instead we have to devise an appropriate bihamiltonian
reduction scheme similar to the one set up in \cite{bfmop10} for the case
\(r=1\). Namely, we should reduce the bihamiltonian manifold consisting of the
linear space \(\Rep_{\bb{R}}(\ol{Q},(n,1))\) equipped with the two compatible
Poisson structures \(\check{\pi}_{0}\) and \(\check{\pi}_{1}\) to a suitable
\(4n\)-dimensional bihamiltonian manifold, which then must be identified with
the usual phase space of the rank 2 Gibbons-Hermsen system.

This is a non-trivial problem that we are not going to tackle here. However,
let us briefly sketch a possible way to construct such a quotient. Following
\cite{bfmop10} it is natural to look for a two-step projection,
\[ \Rep_{\bb{R}}(\ol{Q},(n,1)) \longrightarrow \mathcal{P} \longrightarrow
\bb{R}^{4n}\,, \]
where the first step involves the definition of a \((n^{2}+4n)\)-dimensional
slice \(\mathcal{P}\) for the action of \(\GL_{n}(\bb{R})\) on the \((2n^{2} +
4n)\)-dimensional space \(\Rep_{\bb{R}}(\ol{Q},(n,1))\). Once this slice has
been defined, the second projection may again be performed by the following procedure: 
\begin{enumerate}
\item[1)] one selects a set of \(4n\) regular \(\GL_{n}(\bb{R})\)-invariant
  functions on \(\Rep_{\bb{R}}(\ol{Q},(n,1))\) which span a Poisson subalgebra
  with respect to both brackets and whose Jacobian matrix with respect to the
  reduced Gibbons-Hermsen coordinates is nondegenerate;
\item[2)] one takes the submersion \(\mathcal{P}\to \bb{R}^{4n}\) whose
  components are given by those functions.
\end{enumerate}
Such a set of generators may consist, for example, of the \(2n\) Hamiltonians
\(\hat{I}_{k}\), \(\hat{I}^{(2)}_{k}\) considered before supplemented with the
\(n\) functions
\[ \hat{J}_{\ell}\deq \tr X^{\ell-1}Y \qquad (1\leq \ell\leq n)\,, \]
familiar from the Calogero-Moser case, and with the further \(n\) functions
\(\hat{J}^{(2)}_{0}, \dots, \hat{J}^{(2)}_{n-1}\), where
\[ \hat{J}^{(2)}_{\ell}\deq \hat{H}_{\ell,e_{12}}\,, \qquad
e_{12} =
\begin{pmatrix}
  0 & 1\\
  0 & 0
\end{pmatrix}\,, \]
which are needed to recover the additional degrees of freedom contained in the
matrices \(v\) and \(w\).

\section{Final remarks}
\label{sec:fin}

We believe that the formalism presented in this paper may be of help in
finding bihamiltonian structures for many other classical finite-dimensional
integrable systems. The most obvious case to be considered next is that of
Calogero-Moser systems with trigonometric/hyperbolic potentials and their
generalizations with internal degrees of freedom (see for instance the survey
\cite{nek99}). In this connection let us observe that, as pointed out by
Bielawski \cite[Remark 7.3]{bie13}, the Poisson bivector $\pi_1$ given in
eq.~\eqref{eq:pi1-cm} also induces, on a suitable open subset of
\(\Rep_{\bb{k}}(\ol{Q}_{\circ},n)\), the symplectic structure of the
trigonometric Calogero-Moser system. The compatibility between \(\pi_{1}\) and
the Poisson bivector \(\pi_{2}\) given in eq.~\eqref{eq:p2-cm} then suggests
that the bihamiltonian description of this system hinted at in \cite{aaj10}
may be derived from this pair of double Poisson structures.

Another promising source of examples may come from the very general class of
integrable systems arising from the Coulomb branch of the moduli space of
vacua in four-dimensional \(N=2\) supersymmetric gauge theories
\cite{sw94a,sw94b,dw96}. As the referee pointed out to us, many explicit
examples of systems of this kind have been derived, most recently by Dorey
and Zhao \cite{dz15}, starting from elliptic quiver gauge theories.
Interpreting these systems from a noncommutative-geometric point of view seems
to be an interesting problem.

The last issue we would like to mention is related to the notion of
\emph{duality} between integrable systems introduced by Ruijsenaars in
\cite{rui88} and later reinterpreted in terms of the symplectic reduction of
two families of commuting Hamiltonians on a higher-dimensional symplectic
manifold \cite{gn95,fgnr00}. Being a relation between canonical coordinates on
actual phase spaces, the Ruijsenaars duality transformation can be implemented
only at the level of symplectic quotients, and is thus invisible at the
noncommutative level. However, in many cases the data to be provided as input
for the construction (namely the ``big'' phase space, the symplectic form and
the two families of commuting Hamiltonians) can be interpreted in terms of
geometric objects on quiver representation spaces.

A relevant example is provided by the well known duality between the
trigonometric Calogero-Moser(-Sutherland) system and the rational
Ruijsenaars-Schneider system, which was put on a firm geometric basis by
Feh\'er and Klim\v{c}\'{\i}k \cite{fk09b}. The input data for their construction
seem to admit a noncommutative-geometric interpretation. If so, it should be
possible to derive a bihamiltonian structure for both systems starting from
the same noncommutative PN structure pointed out above (see again the related
computations in \cite{aaj10}).

\subsection*{Acknowledgments}

This work was partially supported by the PRIN ``Geometria delle variet\`a
algebriche'' and by the University of Genoa's research grant ``Aspetti
matematici della teoria dei campi interagenti e quantizzazione per
deformazione''. The authors are grateful to the referee, whose suggestions
have been incorporated in section~\ref{sec:fin}. A.T. would like to thank the
Department of Mathematics at the University of Genoa for the kind hospitality
during the period in which this paper was written.

\bigskip
\textsc{Dipartimento di Matematica, Universit\`a di Genova, via Dodecaneso 35, 16146 Genova, Italy}

\emph{Email addresses}: \texttt{bartocci@dima.unige.it}, \texttt{altacch@gmail.com}

\end{document}